\documentclass[twocolumn,pra,aps,superscriptaddress]{revtex4-1}
\usepackage[utf8]{inputenc}  
\usepackage[T1]{fontenc}     %Output what you want e.g., é, ł, a, ü
\usepackage[british]{babel}  %Do hyphenation according to british english
\usepackage[sc,osf]{mathpazo}
\usepackage[scaled=0.86]{berasans}  % URL font that go well wtih palatino
\usepackage[colorlinks=true, citecolor=blue, urlcolor=blue]{hyperref}  %Hyperlinks (pink, green, blue)
\usepackage{graphicx} % Package to insert exteral figures
\usepackage[babel]{microtype}  %Improves text justification
\usepackage{amsmath,amssymb,amsthm,bm,amsfonts,mathrsfs,bbm} %Usefull math packages

\usepackage{xspace}  %Useful to add space in macros
\usepackage{pgfplots}
\usepackage{xcolor,colortbl}
\usepackage{array}
\usepackage{bigstrut}
\usepackage{hyperref}
\usepackage{amsthm}
\usepackage{natbib}
\newtheorem{prop}{Proposition}
\newtheorem{lem}{Lemma}
\newtheorem{thm}{Theorem}

\begin{document}

\title{Strong quantum solutions in conflicting-interest Bayesian games}

%#####################################

\author{Ashutosh Rai}
\email{arai@iip.ufrn.br}

\affiliation{
	International Institute of Physics, Federal University of Rio Grande do Norte, 59070-405 Natal, Brazil}
\affiliation{Center for Quantum Computer Science, University of Latvia, Raina Bulv. 19, Riga, LV-1586, Latvia}

\author{Goutam Paul}
\email{goutam.paul@isical.ac.in}
\affiliation{Cryptology and Security Research Unit, R. C. Bose Centre for Cryptology and Security, Indian Statistical Institute, Kolkata 700108, India}

%########################################

\begin{abstract}
Quantum entanglement has been recently demonstrated as a useful resource in conflicting interest games of incomplete information between two players, Alice and Bob \href{http://journals.aps.org/prl/abstract/10.1103/PhysRevLett.114.020401}{[Pappa \emph{et al.}, Phys. Rev. Lett. {\bf 114}, 020401 (2015)]}. General setting for such games is that of correlated strategies where the correlation between competing players is established through a trusted common adviser; however, players need not reveal their input to the adviser. So far, quantum advantage in such games has been revealed in a restricted sense. Given a quantum correlated equilibrium strategy, one of the players can still receive a higher than quantum average payoff with some classically-correlated equilibrium strategy. In this work, by considering a class of asymmetric Bayesian games, we show the existence of games with quantum correlated equilibrium where average payoff of both the players exceed respective individual maximum for each player over all classically-correlated equilibriums. 
\end{abstract}

%\pacs{03.65.Ud, 02.50.Le, 03.67.Ac}

\maketitle

%########################################

\section{Introduction}
Quantum entanglement is an invaluable resource for information processing tasks \cite{ben1,ben2,hor}. Entanglement gives rise to the phenomenon of quantum nonlocality, which shows that non-communicating spatially separated parties, by presharing parts of an entangled quantum state, can generate correlations which cannot be explained by any local realistic theory \cite{epr,Bell'1964,nl}. Several applications of nonlocal correlations exists, for example, in quantum cryptography \cite{eka} and nonlocal games of full cooperation \cite{cleve}. However, applications of quantum nonlocality to achieve the best possible solution in Bayesian games of conflicting interests is a very recent development \cite{Pappa,Roy'2016,Auletta'2016}.

Games usually model the situation of some conflict between a given number of parties \cite{osb}. An interesting connection between Bell nonlocality \cite{Bell'1964} (which can be realized in quantum mechanics through entangled states) and Bayesian games introduced by Harsanyi \cite{Harsanyi'1967} was established in \cite{Brunner'2013}. Soon after, an explicit example of a two-party game with conflicting interests where an entangled state leads to a better solution was provided in \cite{Pappa} and has inspired a number of interesting works along this direction \cite{Roy'2016,Auletta'2016,Banik'2017,citu1,citu2,ramon,lason}. There had also been previous attempts to devise quantum strategies for \emph{non}-Bayesian games, providing an advantage under certain specific restrictions \cite{eisert1}; however, the physical applicability of such results had been debated \cite{benj,eisert2,van}. Subsequent progress made in Refs. \cite{Brunner'2013,Pappa,Roy'2016,Auletta'2016,Banik'2017,citu1,citu2,ramon,lason,iqbal} is essentially along the direction outlined in Ref. \cite{van}.

In a Bayesian game the competing parties have only partial information about the whole setting in which the game is played \cite{Harsanyi'1967}. For example, each player may have some random private information, such as her (his) input, unknown to other players. In general, to resolve the conflict in the best possible way, players can arrange a common trusted adviser facilitating correlated strategies \cite{Aumann'1987}. By resolution of conflicts it means all players agree to adopt a strategy which is a Nash equilibrium \cite{nash}. 

\emph{Nash equilibrium.} The Nash equilibrium is defined as follows: In an $n$ players game, let $\mathbf{S}$ be the collection of all possible global strategies derived from each player's local strategies. Let $s^*_i$ for $i\in\{1,...,n\}$ be some local strategy of the $i$th player. Then a global strategy $s^*=(s^*_1,...,s^*_n)\in \mathbf{S}$ is a Nash equilibrium when, if any one player, say $j$th, fixes her or his local strategy to $s^*_j$ and then none of the remaining players can gain by making unilateral changes in their local strategies. 

In this work, we demonstrate a qualitatively stronger quantum advantage in Bayesian games of conflicting interests. For this, we introduce a class of asymmetric two-player Bayesian games with conflicting interests. The general setting for these games is that of correlated strategies where a common trusted adviser can provide the two players with some classical correlation or some entangled state as an advice. Input to each player is generated uniformly at random and is fully private information not revealed to the adviser. Payoff to each player in a given round depends on joint input (type) and joint output (action) of both the players. The measure of the reward to each player, on adopting a certain strategy, is the \emph{average payoff} when the game is played for many rounds with that strategy. 

Many real-world problems can be modeled through such games; some interesting examples of such models can be found in Refs. \cite{Pappa,Brunner'2013,Auletta'2016}. Assumption of an uninformed (nevertheless trusted) adviser insures that no information about one player's input can leak to the other player; this assumption has been motivated in \cite{Brunner'2013} as no-signaling advice and in \cite{Auletta'2016} as belief-invariant advice. 
 
Our results can be summarized as follows. We define a one-parameter family of asymmetric games (which includes the game defined in~\cite{Pappa} as a special case). For the considered class of game, we compute and analyze uncorrelated Nash equilibriums. Next, we analyze classically correlated Nash equilibriums for our games and find that the sum of average payoff of two players in these games is equivalent to the Bell--Clauser-Horne-Shimony-Holt (CHSH) expression up to scaling. Then we show that with Popescu-Rohrlich--box ($PR$-box) correlations the achieved Nash equilibrium has a strong feature (also noted in~\cite{Roy'2016}). Finally, we discuss a quantum correlated strategy and show that in a certain parameter range there are quantum Nash equilibriums in which individual rewards to each player exceed their respective average payoff over the set of all possible classically correlated Nash equilibriums. This is the main result of our paper.
 
\section{Class of two party Bayesian games}
Consider a class of Bayesian games, $\{G(\epsilon): 0\leq \epsilon\leq 3/4\}$, played between two spatially separated players, Alice and Bob. In each single round of the game, both the players receive some type (input) from the set $\{0,1\}$. No player has any information about the type (input) of the other player. Let types of Alice and Bob be denoted by $x_A$ and $x_B$, respectively, and suppose that joint input $x=(x_A,x_B)$ is sampled from a uniform distribution. After receiving their types, players perform some action $y=(y_A,y_B)$ where $y_A$ is an action (output) of Alice and $y_B$ is an action of Bob. The utility (payoff) function, $u_i(x,y)$ where $i\in\{A,B\}$, of the game is given in Table (\ref{tab1}), where we use the symbol $A$ ($B$) as a short hand for Alice (Bob). 
%\begin{widetext}
\begin{center}
\begin{table}[h]
\centering
\caption{Utility functions of Alice and Bob for the considered class of games $G(\epsilon)$ where $0\leq \epsilon \leq 3/4$.}
\begin{tabular}{cc|c}
\hline \hline
    & ~~~~~~$\mathbf{x_A\wedge x_B=0}$ &  ~~~~~~$\mathbf{x_A\wedge x_B=1}$ \\ \hline
 & $~~~~~~~~\mathbf{y_B=0} ~~~~~~~~~~\mathbf{y_B=1}$~~~~&~~~~$\mathbf{y_B=0}~~~~~~~~~~\mathbf{y_B=1}$ \\ 

 $\mathbf{y_A=0}$& $(1-\epsilon,~~\frac{1}{2}+\epsilon)~~~~(~~0,~~0~~)$ & ~~$(~~0,~~~~~~0~~)~~~~~(~\frac{3}{4},~\frac{3}{4}~)$\\

 $\mathbf{y_A=1}$& $~~(~~0~,~~~~~~0~~)~~~~~(~~\frac{1}{2},~~1~~)$ & $(\frac{3}{4}-\epsilon,~\frac{3}{4}+\epsilon) ~~~~(~0,~0~)$\\
 
 \hline \hline
\end{tabular}
\label{tab1}
\end{table}
\end{center}
%\end{widetext}

In Table (\ref{tab1}), utilities of Alice and Bob are ordered as $(u_A,u_B)$. Note that for achieving high payoffs both players need to correlate their actions when $x_A\wedge x_B=0$ and anticorrelate their actions when $x_A\wedge x_B=1$.
When $x_A\wedge x_B=0$: (i) for $0\leq \epsilon< 1/2$, Alice receives a higher payoff when the correlated action of Alice and Bob is $(0,0)$ and Bob receives a higher payoff when the correlated action of Alice and Bob is $(1,1)$; (ii) for $\epsilon=1/2$, correlating either way gives the same payoff to both the players; and
(iii) for $1/2<\epsilon \leq 3/4$, Alice benefits from $(1,1)$ and Bob from $(0,0)$ correlation. On the other hand, when $x_A\wedge x_B=1$, Alice prefers $(0,1)$ and Bob prefers $(1,0)$ correlated outcome for the entire range of $\epsilon$. 

The average payoff that each player receives in games $G(\epsilon)$ after a large number of rounds can be computed from Table (\ref{tab1}) and can be expressed as follows:
\begin{eqnarray}
\left\langle u_A\right\rangle &=& \frac{1}{4}\left\{(1-\epsilon) P(00|00)+ \frac{1}{2}P(11|00)\right\}\nonumber\\&+&\frac{1}{4}\left\{ (1-\epsilon)P(00|01)+ \frac{1}{2}P(11|01)\right\}\nonumber \\ &+&\frac{1}{4}\left\{(1-\epsilon) P(00|10)+ \frac{1}{2}P(11|10)\right\}\nonumber \\ &+&\frac{1}{4}\left\{(\frac{3}{4}-\epsilon) P(10|11)+ \frac{3}{4} P(01|11)\right\}. \label{fa}
\end{eqnarray} 

\begin{eqnarray}
\left\langle u_B\right\rangle &=& \frac{1}{4}\left\{ (\frac{1}{2}+\epsilon)P(00|00)+ P(11|00)\right\}\nonumber\\&+&\frac{1}{4}\left\{(\frac{1}{2}+\epsilon) P(00|01)+ P(11|01)\right\}\nonumber \\ &+&\frac{1}{4}\left\{ (\frac{1}{2}+\epsilon)P(00|10)+ P(11|10)\right\}\nonumber \\ &+&\frac{1}{4}\left\{ (\frac{3}{4}+\epsilon)P(10|11)+ \frac{3}{4} P(01|11)\right\}.
\label{fb}
\end{eqnarray} 

 We conclude this section by noting the following: (i) The game $G(0)$ corresponding to $\epsilon=0$ is the symmetric conflicting-interest game introduced and analyzed by Pappa and co-authors in \cite{Pappa}. (ii) For $\epsilon \neq 0$, all the games $G(\epsilon)$ are asymmetric and Bob has some advantage over Alice. (Such situations can occur, for instance, in example $3$ of the game discussed in \cite{Brunner'2013}, when the two competing companies $A$ and $B$ can be of different size.)

\section{Uncorrelated Nash equilibrium}
Consider all the possible pure (deterministic) strategies for playing the game. Then, each player can choose from the following four strategies: 
\begin{eqnarray}
S_1^{i}&:& y_i=0, \nonumber \\
S_2^{i}&:& y_i=1, \nonumber \\
S_3^{i}&:& y_i=x_i, \nonumber \\
S_4^{i}&:& y_i=x_i\oplus 1, \label{eq3}
\end{eqnarray}
where $x_i$ is the uniformly random local input bit of the $i$th player and $y_i$ is the output bit according to any of the four possible deterministic strategies $\{S_1^{i},S_2^{i},S_3^{i},S_4^{i}\}$; here the index $i\in\{A,B\}$ is the label for the two players Alice and Bob. Thus there are a total of $16$ pure strategies for playing the game. Table (\ref{tab2}) lists the average payoff for each player for all possible pure strategies.

From Table (\ref{tab2}) one can easily find all the pure Nash equilibriums for the considered class of games $G(\epsilon)$. In the parameter range $0\leq \epsilon\leq 1/4$ the strategies $(S_1^{A}, S_3^{B})$, $(S_3^{A}, S_4^{B})$, and $(S_4^{A}, S_2^{B})$ are Nash equilibriums; for $1/4\leq \epsilon \leq 1/2$ the strategies $(S_1^{A}, S_1^{B})$, $(S_3^{A}, S_4^{B})$, and $(S_4^{A}, S_2^{B})$ are Nash equilibriums; and for $1/2\leq \epsilon \leq 3/4$, $(S_1^{A}, S_1^{B})$, $(S_2^{A}, S_4^{B})$, and $(S_4^{A}, S_3^{B})$ are Nash equilibriums. For all values of $\epsilon$ there are multiple equilibriums, and one can easily check that the two players will have different preferred equilibrium and therefore there are conflicting interests in this game for all values of $\epsilon$. For example, in the parameter range $0\leq \epsilon <1/4$, the most preferred strategy for Alice is $(S_1^{A}, S_3^{B})$ whereas Bob gives highest preference to $(S_4^{A}, S_2^{B})$.

The two players can also play this game by adopting some mixed strategies. In a mixed strategy Alice (Bob) chooses to implement from four pure strategies given in Eq. (\ref{eq3}) according to some probability distribution. First we note that any pure equilibrium point remains an equilibrium point in the general setting, which extends to all possible mixed strategies. Are there any other uncorrelated mixed Nash equilibriums for this game? In general it is mathematically complex to find all mixed Nash equilibriums, and in the latter part of the paper we will note that it is not necessary to explicitly compute them for our purposes.  

\begin{widetext}
\begin{center}
\begin{table}[h]
\centering
\caption{\label{tab2}Average payoffs of Alice and Bob ordered as $(\left\langle u_A\right\rangle, \left\langle u_B\right\rangle)$ for different
pure strategies for the considered class of games $G(\epsilon)$. In parameter range $0\leq \epsilon\leq 1/4$ the strategies $(S_1^{A}, S_3^{B})$, $(S_3^{A}, S_4^{B})$, and $(S_4^{A}, S_2^{B})$ are Nash equilibriums, marked in bold with round brackets; for $1/4\leq \epsilon \leq 1/2$, $(S_1^{A}, S_1^{B})$, $(S_3^{A}, S_4^{B})$, and $(S_4^{A}, S_2^{B})$ are Nash equilibriums, marked in bold with curly brackets; and for $1/2\leq \epsilon \leq 3/4$, $(S_1^{A}, S_1^{B})$, $(S_2^{A}, S_4^{B})$, and $(S_4^{A}, S_3^{B})$ are Nash equilibriums, marked in bold with square brackets.}
\begin{tabular}{c|cccc}
 \hline \hline \\
  A/B     & $S_1^{B}$ &  $S_2^{B}$ ~~~~& $S_3^{B}$ &~~~~ $S_4^{B}$\\ 
  \hline \\
	
	$S_1^{A}$ & \boldmath $\left[\left\{\frac{3}{4}-\frac{3\epsilon}{4},\frac{3}{8}+\frac{3\epsilon}{4}\right\}\right]$~~~~&  $(\frac{3}{16},\frac{3}{16})$ ~~~~& \boldmath $\left(\frac{11}{16}- \frac{\epsilon}{2},\frac{7}{16}+ \frac{\epsilon}{2}\right)$ ~~~~&$(\frac{1}{4}-\frac{\epsilon}{4},\frac{1}{8}+\frac{\epsilon}{4})$ \\  \\

$S_2^{A}$ & $(\frac{1}{8},\frac{1}{4})$~~~~&  $(\frac{3}{8},\frac{3}{4})$ ~~~~& $(\frac{1}{8},\frac{1}{4})$ ~~~~&\boldmath$\left[\frac{7}{16}-\frac{\epsilon}{4},\frac{11}{16}+\frac{\epsilon}{4}\right]$\\  \\

 $S_3^{A}$&$ (\frac{11}{16}-\frac{3\epsilon}{4},\frac{7}{16}+\frac{3\epsilon}{4})$~~~~&  $(\frac{1}{8},\frac{1}{4})$~~~~ & $(\frac{1}{4}-\frac{\epsilon}{4},\frac{1}{8}+\frac{\epsilon}{4})$ ~~~~&\boldmath $\left\{\left( \frac{9}{16}-\frac{\epsilon}{2},\frac{9}{16}+\frac{\epsilon}{2}\right)\right\}$\\ \\

 $S_4^{A}$& $(\frac{1}{4}-\frac{\epsilon}{4},\frac{1}{8}+\frac{\epsilon}{4})$~~~~&  \boldmath$\left\{\left(\frac{7}{16},\frac{11}{16}\right)\right\}$~~~~ & \boldmath$\left[\frac{9}{16}-\frac{\epsilon}{4},\frac{9}{16}+\frac{\epsilon}{4}\right]$ ~~~~&$(\frac{1}{8},\frac{1}{4})$\\ \\
 
 \hline \hline
\end{tabular}
\end{table}
\end{center}
\end{widetext}

\section{Classical Correlated strategies}
In a correlated strategy for the game, a trusted common adviser is introduced. It is assumed, as in \cite{Pappa}, that the adviser has no information about types that each player receive in any given round, meaning that the type of each player in any given round is fully private. Therefore, the adviser can send advice even before the start of every round of the game. Any classically correlated advice is simply a set of separate instructions to both players to implement a strategy $(S_i^A, S_j^B)$ according to some probability distribution $\{p_{ij}: i,j\in\{1,2,3,4\}\}$. The advice from the adviser is like a preshared correlation between the two players before receiving their types (for example, preshared randomness). It is well known that such an assumption leads to bounds on certain functions of joint conditional probabilities of player outputs which are known as Bell-type inequalities.

It turns out that the class of games that we consider has a nice connection with the Bell-CHSH inequality 
 $\mathbb{B}=\langle A_0B_0\rangle+\langle A_0B_1 \rangle+\langle A_1B_0 \rangle-\langle A_1B_1\rangle \leq 2 $ \cite{chsh}. It can be shown (see Appendix A for steps in the derivation) that the sum of average payoff of the two players is related to the value of the Bell-CHSH expression as $\left\langle u_A\right\rangle+\left\langle u_B\right\rangle=\frac{3}{16}(\mathbb{B}+4)$. Then from the local bound on the Bell-CHSH expression it immediately follows that
 \begin{equation}
 \left\langle u_A\right\rangle_{C}+\left\langle u_B\right\rangle_{C}\leq \frac{9}{8}, \label{fab} 
 \end{equation}
 where the subscript $C$ stands for \emph{classical}.
 Thus when the game is played with any classically correlated strategy the sum of average payoff of the two players is bounded as given by Eq. (\ref{fab}). 
 
On the other hand, the average payoff of the individual player Alice (Bob) for any classically correlated Nash equilibrium is bounded by the maximum of all average payoffs of Alice (Bob) when the game is played with uncorrelated pure strategies. These bounds can be computed from Table (\ref{tab2}) and are as follows: 

\begin{eqnarray}
    \left\langle u_A\right\rangle_{C}\leq
\begin{cases}
    \frac{3}{4}(1-\epsilon)& \text{for } ~0\leq \epsilon \leq \frac{1}{4},\\
    \frac{11}{16}-\frac{\epsilon}{2}              & \text{for }~ \frac{1}{4}< \epsilon \leq\frac{1}{2},\\
    \frac{7}{16}              &\text{for }~ \frac{1}{2}< \epsilon \leq \frac{3}{4}.
\end{cases}\label{eq7}
\end{eqnarray}

\begin{eqnarray}
    \left\langle u_B\right\rangle_{C}\leq
\begin{cases}
    \frac{3}{4}& \text{for }~ 0\leq \epsilon \leq \frac{1}{4},\\
    \frac{11}{16}+\frac{\epsilon}{4} & \text{for } ~\frac{1}{4}< \epsilon \leq\frac{1}{2},\\
    \frac{7}{16}+\frac{3\epsilon}{4}              &\text{for } ~\frac{1}{2}< \epsilon \leq \frac{3}{4}.
\end{cases}\label{eq8}
\end{eqnarray}

Note that these bounds for $\left\langle u_A\right\rangle_C$ and $\left\langle u_B\right\rangle_C$ may not be tight bounds on an average payoff for correlated Nash equilibriums. However, we will see in the following section that these bounds are sufficient to reveal a feature of strategies which use PR-box correlations in which they outperform all classically correlated strategies in the parameter range $ 0\leq \epsilon \leq 5/8$. Incidentally, for a different class of conflicting interest games, such a feature of strategies which use PR-box correlation was also revealed in a recent work \cite{Roy'2016}. 
 
\section{Nonlocal PR-Box strategy}
Suppose Alice and Bob share the $PR$-box correlation \cite{pr} provided to them as an advice from the adviser. These correlations are the extremal no-signaling nonlocal correlation with two input and two output bits. Although known not to exist in nature, these correlations are widely used as a conceptual tool to gain insight into plausible applications of physically assessable nonlocal quantum correlations.
 
\emph{Nonlocal strategy $PR^{*}$}---Alice (Bob) feed type $x_A$ ($x_B$) as an input bit to her (his) part of the $PR$-box and output bit $y_A$($y_B$) is the action bit. Then the joint probabilities of the player's actions given their types must satisfy

\begin{eqnarray}
    P(y_A,y_B|x_A,x_B)=
\begin{cases}
    \frac{1}{2}& \text{if }~ y_A\oplus y_B= x_A\wedge x_B,\\
    0 & \text{otherwise}. 
\end{cases}
\end{eqnarray}

With this strategy the average payoffs of the two players are as follows:
\begin{eqnarray}
\left\langle u_A\right\rangle_{PR^*}= \frac{1}{2}\left(\frac{3}{2}-\epsilon \right),\\
\left\langle u_B\right\rangle_{PR^*}= \frac{1}{2}\left(\frac{3}{2}+\epsilon \right).
\end{eqnarray}
\begin{prop}
The nonlocal strategy $PR^{*}$ is a Nash equilibrium for all $G(\epsilon)$ in the parameter range $ 0\leq \epsilon \leq 5/8$.
 \end{prop}
 
 \begin{proof}
 Suppose Alice deviates from the strategy $PR^{*}$, whereas Bob's strategy is fixed. In general, Alice can deviate from the strategy by preprocessing the received type and feeding the result as an input to her end of the PR-box, and answer after some postprocessing of the outcome from the PR-box. Any such local preprocessing and postprocessing can lead to some new joint probability distribution $\tilde{P}(y_A,y_B|x_A,x_B)$ which can be different from the PR-box correlation given by the Eq. (7). However, it follows from the no-signaling principle that $\tilde{P}(y_A,y_B|x_A,x_B)$, obtained by local processing, must belong to the set of no-signaling distribution. The same argument holds when Alice keeps her strategy fixed and Bob deviates.
 
 Now the idea of the proof is that in the parameter range $ 0\leq \epsilon \leq 5/8$, it is sufficient to maximize $\left\langle u_A\right\rangle$ and $\left\langle u_B\right\rangle$, given by Eqs. (1) and (2) respectively, over the set of all no-signaling correlations. The set of no-signaling joint probability distributions for two parties with binary inputs and outputs is fully characterized by Barrett \emph{et al.} in \cite{bar}. The set forms an eight-dimensional polytope with $24$ vertices---$16$ local vertices and $8$ nonlocal vertices. Therefore, the maximum value of any linear function of joint probabilities $P(y_A,y_B|x_A,x_B)$ over the set of no-signaling correlations is achieved at some vertex of the polytope. 
 	
 By computing the value of $\left\langle u_A\right\rangle$, given by Eq. (1), for all $24$ vertices, we find that $\left\langle u_A\right\rangle_{max}=\frac{1}{2}\left(\frac{3}{2}-\epsilon \right)$ and this maximum value is achieved at the PR-box correlation given by the Eq. (7). This proves that Alice cannot increase her average payoff by any unilateral deviation. For the case when Alice keeps her strategy fixed and Bob deviates, by computing the value of $\left\langle u_B\right\rangle$, given by Eq. (2), for all $24$ vertices, we find that $ \left\langle u_B\right\rangle_{max}=\frac{1}{2}\left(\frac{3}{2}+\epsilon \right)$ and is achieved at the PR-box correlation given by the Eq. (7). This proves that Bob cannot increase his average payoff by any unilateral deviation.
 \end{proof}
 From the bounds given by conditions (\ref{eq7}) and (\ref{eq8}) it follows that the $PR^{*}$ strategy has a strong property that $\left\langle u_A\right\rangle_{PR^*}> \mbox{Max}~(\left\langle u_A\right\rangle_{C})$ and $\left\langle u_B\right\rangle_{PR^*}> \mbox{Max}~(\left\langle u_B\right\rangle_{C})$ over all possible classically correlated strategies and for all $\{G(\epsilon): 0\leq \epsilon \leq 5/8\}\}$. This result indicates that a similar feature may exist when $G(\epsilon)$ is played with quantum entanglement as an advice. In the following section we show that indeed this feature also holds for quantum strategies, albeit for a smaller range of the parameter $\epsilon$.

\section{Quantum strategy}
Consider a quantum protocol for playing $G(\epsilon)$, which is the well-known protocol for achieving the maximum violation for the Bell-CHSH inequality \cite{cleve}.
We give essentially the same protocol as in \cite{cleve}, but the shared entangled state and measurements by the players are slightly modified (rotated).

\emph{Quantum strategy $Q^{*}$.} Let the singlet state $|\Psi_{AB}^{-}\rangle=\frac{1}{\sqrt{2}}(|01\rangle-|10\rangle)$ be the quantum advice shared between Alice and Bob. The strategy of Alice is to measure $\sigma_z$ on receiving the type $0$ and $\sigma_x$ on receiving type $1$, whereas Bob performs measurement $-\frac{1}{\sqrt{2}}(\sigma_x+\sigma_z)$ on receiving $0$ and measurement $\frac{1}{\sqrt{2}}(\sigma_x-\sigma_z)$ on receiving $1$. The action of Alice and Bob is to answer with $0$ ($1$) when the $+1$ ($-1$) eigenstate clicks. With this quantum strategy for playing the game, the average payoffs of Alice and Bob are as follows:

\begin{eqnarray}
\left\langle u_A\right\rangle_{Q^*}= \frac{1}{4}\left(1+\frac{1}{\sqrt{2}}\right)\left(\frac{3}{2}-\epsilon \right),\\
\left\langle u_B\right\rangle_{Q^*}= \frac{1}{4}\left(1+\frac{1}{\sqrt{2}}\right)\left(\frac{3}{2}+\epsilon\right).
\end{eqnarray}
The sum of the average payoff $\left\langle u_A\right\rangle_{Q^*}+\left\langle u_B\right\rangle_{Q^*}= \frac{3}{4}\left(1+\frac{1}{\sqrt{2}}\right)$, and this value corresponds to the maximum Bell-CHSH violation in quantum mechanics \cite{cirelson}.

\begin{thm}
The quantum strategy $Q^*$ is a Nash equilibrium for all $G(\epsilon)$.
\end{thm}
\begin{proof}
The main steps in the proof are as follows; details of calculation are provided in Appendix B. First we consider that Alice unilaterally changes her local measurements and Bob keeps his measurements fixed to those in the quantum strategy $Q^*$. On receiving types $0$ and $1$, Alice respectively performs some two-outcome positive operator-valued measure (POVM) measurement $\{X_{00}+X_{01}=\mathbb{I}\}$ and $\{X_{10}+X_{11}=\mathbb{I}\}$, where $X_{00}$, $X_{01}$, $X_{00}$, and $X_{01}$ are positive operators acting on the two-dimensional complex Hilbert space $\mathbb{C}^2$. In each case Alice answers with the measurement outcome. We then derive the expression for maximum average payoff of Alice over all possible POVM measurements and find that the maximum is achieved with the strategy $Q^*$. This implies that Alice cannot increase her average payoff by deviating unilaterally from the quantum strategy $Q^*$. We then prove the same result when Alice's measurements are fixed and Bob changes his measurements to all possible POVMs. The calculations involved in the proof are provided in Appendix B. Finally, since we have maximized over all possible unilateral changes in one party measurement, it is easy to check that the results hold even under any preprocessing of inputs or postprocessing of measurement outcomes. 
\end{proof}
The quantum strategy $Q^*$, unlike the $PR^*$ strategy, cannot beat the bounds on average payoffs for Alice and Bob given by conditions (\ref{eq7}) and (\ref{eq8}). Interestingly, however, by making one of these bounds tighter, as derived in the following lemma, we can recover the strong property shown by the $PR^*$ strategy for a subset of the games $G(\epsilon)$.
\begin{lem}
In the parameter range $1/4\leq \epsilon \leq 1/2$, no classical-correlated equilibrium exists where Bob can play the strategy $S_3^{B}$ with a nonzero probability.\label{lem1}
\end{lem}
\begin{proof}
In a classically correlated strategy players are advised to implement strategy $(S_i^A, S_j^B)$ according to some probability distribution $\{p_{ij}: i,j\in\{1,2,3,4\}\}$. This would mean that Alice adopts strategy $S_i^A$ with probability $\mu_i=\sum_{j=1}^{4}p_{ij}$, and Bob plays strategy  $S_j^B$ with probability $\lambda_j=\sum_{i=1}^{4}p_{ij}$. Now, for the parameter range $1/4\leq \epsilon \leq 1/2$, from Table (\ref{tab2}), one can easily verify that any probability distribution $p_{ij}$ with $\lambda_3 \neq 0$ (i.e., Bob playing $S_3^B$ with a non zero probability) cannot be a classically correlated equilibrium. This is due to the fact that if Alice plays with such an advice (where $\lambda_3 \neq 0$), Bob can always deviate from his recommended strategy and increase his average payoff.
\end{proof}

\begin{thm}
The quantum strategy $Q^*$ gives a strong advantage if ~$\frac{3}{14}(3-\sqrt{2})\leq \epsilon \leq \frac{1}{4}(-8+7\sqrt{2})$,~ in the sense that, in this range, both the players beat their individual maximum average payoff over all possible classically correlated Nash equilibriums.\label{thm2}
\end{thm}
\begin{proof}
Using Lemma (\ref{lem1}), we get a tighter bound on the average payoff of Alice for classically correlated equilibrium, which is as follows:
$$\left\langle u_A\right\rangle_{C}\le \mbox{Max}\left\{\frac{7}{16},~\frac{3}{4}-\frac{3\epsilon}{4},\right\}~~\mbox{for}~~\frac{1}{4}\leq \epsilon \leq \frac{1}{2}.$$
For Bob, in the parameter range $1/4\leq \epsilon \leq 1/2 $ we apply the same bound as given in the condition (\ref{eq8}). From these bounds, we find that $\left\langle u_A\right\rangle_{Q^*}>\left\langle u_A\right\rangle_{C}$ and $\left\langle u_B\right\rangle_{Q^*}>\left\langle u_B\right\rangle_{C}$ if $c_1\leq \epsilon \leq c_2$ where $c_1=\frac{3}{14}(3-\sqrt{2})\approxeq 0.34$ and $c_2=(-8+7\sqrt{2})/4 \approxeq 0.47$.
\end{proof}
Theorem (\ref{thm2}) shows that in the class of conflicting-interest Bayesian games that we consider, there are several examples of games where quantum entanglement gives a much stronger advantage than any such game proposed so far \cite{Pappa,Roy'2016}. For example, $G(2/5)$, $G(3/8)$, etc. have the property that the individual average payoff of both players with the quantum strategy $Q^*$ exceeds the maximum average payoff that each player can achieve in any classically correlated equilibrium; this property does not hold, for example, for the game $G(0)$ proposed in \cite{Pappa}.

\section{Concluding remarks}
To conclude, in this work, by designing a class of two-player conflicting-interest Bayesian games $\{G(\epsilon): 0\leq \epsilon \leq 3/4\}$, we show the existence of games where entanglement acts as a more powerful resource than those discovered so far \cite{Pappa,Roy'2016}, leading to strong quantum solution(s). Our work also gives several examples of unfair quantum Nash equilibriums conjectured in \cite{Roy'2016}. Interestingly, for a given quantum advice (an entangled state), there are unfair quantum Nash equilibrium solutions, which also gives the optimal social welfare solutions discussed in \cite{Roy'2016}.

In contrast to the classically correlated advice, a quantum setting for the games which provides entanglement as quantum advice and gives freedom to arbitrarily choose measurement settings to both the players is a weaker quantum advice and can still give strong quantum solutions. It will be reasonable to consider in future works an even stronger version of quantum advice where the distributed quantum correlation is exactly specified, i.e., where the adviser provides both the entangled state and the measurements to the players as black boxes generating exactly one quantum correlation (i.e., a single quantum probability distribution). In such frameworks it is quite likely that the quantum social welfare solutions presented in \cite{Roy'2016} also turn out to be Nash equilibrium.

The key quantum feature which leads to better quantum Nash equilibriums is the existence of nonlocal correlations in the quantum (physical) world. Since all bipartite pure entangled states demonstrate this feature \cite{gisin}, it may hold that any such a quantum state will give better quantum Nash equilibrium in some Bayesian games of conflicting interests.

\ \\

\noindent {\bf Acknowledgements.}
Part of the work was done during the first author's visit to the R. C. Bose Centre for Cryptology and Security, Indian Statistical Institute Kolkata. Fruitful discussions with the authors of the work \cite{Roy'2016} is thankfully acknowledged. This work was supported in part by the European Union Seventh Framework Programme (FP7/2007-2013) under the RAQUEL (Grant Agreement No. 323970) project, QALGO (Grant Agreement No. 600700) project, the ERC Advanced Grant MQC, and  the Brazilian ministries MEC and MCTIC. The authors also take this opportunity to thank the anonymous reviewer(s) for the editorial and technical feedback that helped to achieve a more coherent presentation.

\appendix
\begin{widetext}
\section{Relation between the games  \boldmath{$ G(\epsilon)$} and Bell-CHSH inequality}
Consider the following Bell-CHSH inequality:
\begin{equation}
\mathbb{B}=\langle A_0B_0\rangle+\langle A_0B_1 \rangle+\langle A_1B_0 \rangle-\langle A_1B_1\rangle \leq 2. \label{bi2}
\end{equation}
For $i,j\in \{0,1\}$, $A_{i},B_{j}$ takes a value $\pm 1$ and are results (outputs) of observables of the two parties, respectively indexed by $i$ and $j$.
In the language of our game, it can be thought that $i,j$ are the types (input) of the two players, and $A_i,B_j \in\{\pm 1\}$ are their respective actions (output). The expected values of product of outcomes are
\begin{eqnarray}
\langle A_{i}B_{j}\rangle =P(-1,-1|i,j)+P(+1,+1|i,j) -P(-1,+1|i,j)-P(+1,-1|i,j).\nonumber 
\end{eqnarray}
We do a following relabeling of the outcomes, $-1\mapsto 0$ and  $+1\mapsto 1$, and then we get
\begin{eqnarray}
\langle A_iB_j\rangle= P(0,0|i,j)+P(1,1|i,j)-P(0,1|i,j)-P(1,0|i,j).
\label{exp}
\end{eqnarray}
Upon inserting Eq. (\ref{exp}) in the Bell-CHSH expression $\mathbb{B}=\langle A_0B_0\rangle+\langle A_0B_1 \rangle+\langle A_1B_0 \rangle-\langle A_1B_1\rangle$, using normalization conditions for probabilities, and from the expressions for $\left\langle u_A\right\rangle$ and $\left\langle u_B\right\rangle$ respectively given by Eqs. (\ref{fa}) and (\ref{fb}) one can obtain that $\left\langle u_A\right\rangle+\left\langle u_B\right\rangle=\frac{3}{16}(\mathbb{B}+4)$.

\section{Complete proof of the Theorem 1.}
In the quantum strategy $Q^*$ the state shared between Alice and Bob is the two-qubit singlet state:
\begin{equation}
|\Psi_{AB}^{-}\rangle=\frac{1}{\sqrt{2}}(~|01\rangle-|10\rangle~).\label{eqa1}
\end{equation}
Suppose one player's measurement strategy is fixed to that of $Q^*$. Let the other player deviate from $Q^*$, which he or she can do by choosing two arbitrary POVM measurements on receiving the respective types $0$ and $1$, which in general can be respectively expressed by $M_0: X_{00}+X_{01}=\mathbb{I}$ and $M_1: X_{10}+X_{11}=\mathbb{I}$, where
\begin{eqnarray} 
X_{00}= \frac{1}{2}\left( \begin{smallmatrix} a_0+a_3&a_1-i a_2\\a_1+i a_2&a_0-a_3\end{smallmatrix} \right),~~~~\mbox{where}~ a_0,a_1,a_2,a_3 \in \mathbb{R}~~\mbox{and}~~\sqrt{a_1^2+a_2^2+a_3^2}\leq a_0 \leq 2-\sqrt{a_1^2+a_2^2+a_3^2},\label{eqa2}\\
X_{10}= \frac{1}{2}\left( \begin{smallmatrix} b_0+b_3&b_1-i b_2\\b_1+i b_2&b_0-b_3\end{smallmatrix} \right),~~~~\mbox{where} ~b_0,b_1,b_2,b_3 \in \mathbb{R}~~\mbox{and}~~\sqrt{b_1^2+b_2^2+b_3^2}\leq b_0 \leq 2-\sqrt{b_1^2+b_2^2+b_3^2}. \label{eqa3}
\end{eqnarray}
Let us denote $\sqrt{a_1^2+a_2^2+a_3^2}=||\vec{a}||$ and $\sqrt{b_1^2+b_2^2+b_3^2}=||\vec{b}||$. Note that $0\leq a_0,~b_0,~||\vec{a}||,~||\vec{b}||\leq1 $, and for all $i\in\{1,2,3\}$, $a_i\leq ||\vec{a}||$ and $b_i\leq ||\vec{b}||$.

The proof of Theorem 1 is as follows:
\begin{proof}
	Let Bob's strategy be fixed, i.e., his measurements are $-\frac{1}{\sqrt{2}}(\sigma_x+\sigma_z)$ on receiving $0$ and $\frac{1}{\sqrt{2}}(\sigma_x-\sigma_z)$ on receiving $1$, and let Alice's two measurements be $M_0$ and $M_1$. Then on calculating, the analytical expression for Alice's average payoff turns out to be: 
	\begin{equation}
	\left\langle u_A\right\rangle_{Q}=\frac{1}{32}\left[(9-4\epsilon)+\left\{(2-4\epsilon)a_0+(3\sqrt{2}-2\sqrt{2}\epsilon)a_3\right\}+\left\{b_0+(3\sqrt{2}-2\sqrt{2}\epsilon)b_1\right\}\right].
	\end{equation}
	\emph{Case 1.} If $0\leq \epsilon \leq 1/2$, then,
	\begin{eqnarray}
	\left \langle u_A\right\rangle_{Q} &\leq& \frac{1}{32}\left[(9-4\epsilon)+\left\{(2-4\epsilon)(2-||\vec{a}||)+(3\sqrt{2}-2\sqrt{2}\epsilon)||\vec{a}||\right\}+\left\{(2-||\vec{b}||)+(3\sqrt{2}-2\sqrt{2}\epsilon)||\vec{b}||\right\}\right] \nonumber \\
	&=& \frac{1}{32}\left[(15-12\epsilon)+\left\{(-2+3\sqrt{2})+(4+2\sqrt{2})\epsilon\right\} ||\vec{a}||+\left\{(3\sqrt{2}-1)-2\sqrt{2}\epsilon\right\} ||\vec{b}||\right] \nonumber \\
	&\leq& \frac{1}{32}\left[(15-12\epsilon)+\left\{(-2+3\sqrt{2})+(4+2\sqrt{2})\epsilon\right\} +\left\{(3\sqrt{2}-1)-2\sqrt{2}\epsilon\right\}\right] \nonumber \\
	&=& \frac{1}{4}\left(1+\frac{1}{\sqrt{2}}\right)\left(\frac{3}{2}-\epsilon \right). \nonumber
	\end{eqnarray}
	\emph{Case 2.} If $1/2<\epsilon \leq 3/4$, then,
	\begin{eqnarray}
	\left \langle u_A\right\rangle_{Q} &\leq& \frac{1}{32}\left[(9-4\epsilon)+\left\{(2-4\epsilon)||\vec{a}||+(3\sqrt{2}-2\sqrt{2}\epsilon)||\vec{a}||\right\}+\left\{(2-||\vec{b}||)+(3\sqrt{2}-2\sqrt{2}\epsilon)||\vec{b}||\right\}\right] \nonumber \\
	&=& \frac{1}{32}\left[(11-4\epsilon)+\left\{(2+3\sqrt{2})-(4+2\sqrt{2})\epsilon\right\} ||\vec{a}||+\left\{(3\sqrt{2}-1)-2\sqrt{2}\epsilon\right\} ||\vec{b}||\right] \nonumber \\
	&\leq& \frac{1}{32}\left[(11-4\epsilon)+\left\{(2+3\sqrt{2})-(4+2\sqrt{2})\epsilon\right\} +\left\{(3\sqrt{2}-1)-2\sqrt{2}\epsilon\right\}\right] \nonumber \\
	&=& \frac{1}{4}\left(1+\frac{1}{\sqrt{2}}\right)\left(\frac{3}{2}-\epsilon \right). \nonumber
	\end{eqnarray}
	
	Therefore for any $0\leq\epsilon\leq 3/4$, the maximum value of $\left \langle u_A\right\rangle_{Q}$ over all $M_0$ and $M_1$ turns out to be $\frac{1}{4}\left(1+\frac{1}{\sqrt{2}}\right)\left(\frac{3}{2}-\epsilon \right)$, which is the same as  $\left \langle u_A\right\rangle_{Q^*}$. Hence, Alice cannot increase her average payoff by deviating.
	
	Now let Alice's strategy be fixed, i.e., her measurements are $\sigma_z$ on receiving $0$ and $\sigma_x$ on receiving $1$, and let Bob's two measurements be $M_0$ and $M_1$. Then, upon calculating, the analytical expression for Bob's average payoff turns out to be
	\begin{equation}
	\left\langle u_B\right\rangle_{Q}=\frac{1}{32}\left[15+\left\{(-2+4\epsilon)a_0+(-3-2\epsilon)a_1+(-3-2\epsilon)a_3\right\}+\left\{(-1+4\epsilon)b_0+(3+2\epsilon)b_1+(-3-2\epsilon)b_3\right\}\right].
	\end{equation}
	By similar steps as adopted for maximizing the analytical expression for Alice's average payoff, one can obtain the maximum for Bob. The maximum value of $\left \langle u_B\right\rangle_{Q}$ over all $M_0$ and $M_1$ turns out to be $\frac{1}{4}\left(1+\frac{1}{\sqrt{2}}\right)\left(\frac{3}{2}+\epsilon \right)$, which is same as  $\left \langle u_B\right\rangle_{Q^*}$. Thus Bob cannot increase his average payoff by deviating. Therefore, we can finally conclude that the quantum strategy $Q^*$ is a Nash equilibrium.
\end{proof}
\end{widetext}

%#######################################

\end{document}